\documentclass[%
 reprint,
 superscriptaddress,
 floatfix,
 amsmath,
 amssymb,
 aps,
 prx,
]{revtex4-2}

\usepackage{graphicx}
\usepackage{dcolumn}
\usepackage{amsmath, amssymb, amsthm, bm, amsfonts, mathrsfs, bbm}
\usepackage{dsfont}
\usepackage{physics}
\usepackage[colorlinks=true,linkcolor=blue,citecolor=blue,urlcolor=blue]{hyperref}
\usepackage[caption=false]{subfig}
\usepackage[dvipsnames]{xcolor}

\newtheorem{theorem}{Theorem}

\newtheorem{corollary}{Corollary}

\begin{document}

\preprint{APS/123-QED}

\title{Certifying quantum states without independence assumptions}

\author{Mariana Navarro}
\email{mariana.navarro@icfo.eu}
\affiliation{ICFO - Institut de Ciències Fotòniques, The Barcelona Institute of Science and Technology, 08860 Castelldefels, Barcelona, Spain}
\affiliation{Luxquanta Technologies S.L., Av. Joan Carles I, 30, 1º1ª. 08908 L´Hospitalet de Llobregat, Barcelona, Spain}

\author{Leonardo Zambrano}
\email{leonardo.zambrano@icfo.eu}
\affiliation{ICFO - Institut de Ciències Fotòniques, The Barcelona Institute of Science and Technology, 08860 Castelldefels, Barcelona, Spain}

\begin{abstract}
Standard quantum verification and certification protocols often assume that experimental sources emit independent and identically distributed (i.i.d.) states. In realistic scenarios, however, temporal drift, memory effects, feedback, and correlated noise can violate this assumption, causing standard analyses to underestimate uncertainty and overestimate device performance. Here, we introduce a general framework for quantum verification and certification in which the state produced in each round may depend arbitrarily on the preceding experimental history. Our method provides rigorous confidence intervals for observable expectation values averaged over the experimental trajectory. The framework covers both full verification, in which every state is measured, and spot-checking certification, in which a randomly selected subset is tested to certify the average observable value of the complementary, remaining systems. We apply the framework to importance-sampled Pauli measurements for direct fidelity estimation, energy estimation, and entanglement witnessing, and to classical shadows, allowing multiple observables to be certified simultaneously from the same test data. We illustrate the framework numerically and demonstrate it experimentally through Bell-state certification on a quantum processor.
\end{abstract}

\maketitle

\section{Introduction}

% Motivation for certification/verification
Robust statistical guarantees for quantum state preparation are essential to both foundational physics and quantum technologies \cite{brunner2014bell, eisert2020quantum, pirandola2020advances}. Because quantum measurements are inherently probabilistic, verification and certification protocols infer properties of prepared systems from data collected over a finite sequence of experimental rounds \cite{cramer2010efficient}. The validity of these inferences therefore depends on the source model: an inaccurate model can invalidate confidence intervals and sample-complexity claims, compromising conclusions in applications ranging from quantum cryptography \cite{pirandola2020advances} and simulation \cite{georgescu2014quantum} to quantum optimization \cite{farhi2014quantum}. This motivates methods that remain valid under weaker assumptions on the source and in the presence of a broader range of experimental imperfections and noise processes.

% What we do
To this end,  we develop a general framework for verifying and certifying linear quantum-state properties without assuming independent and identically distributed (i.i.d.)\ state preparations. We model the source as producing a sequence of states $\rho_1,\ldots,\rho_N$, where each $\rho_t$ may depend arbitrarily on the preceding experimental history. In verification, our guarantees concern the trajectory-averaged value $\frac{1}{N}\sum \tr(W\rho_t)$ of an observable $W$. In spot-checking certification, only a randomly selected subset is tested, and the guarantee concerns the corresponding average over the complementary, remaining systems.

% Introducing the idea of iid protocols
Traditionally, quantum state verification and certification protocols---including state tomography \cite{hradil1997quantum, guta2020fast, anshu2024survey}, classical shadows \cite{huang2020predicting}, and fidelity estimation \cite{flammia2011direct, da2011practical, miguel2022collective, riera2023nondestructive, barbera2025sampling, fawzi2026optimal}---assume that the source emits i.i.d.\ quantum states. This means that the same state $\rho$ is prepared in every experimental round without correlations between preparations, so the joint state over $N$ rounds is $\rho^{\otimes N}$. This tensor-product structure justifies applying standard concentration bounds for independent samples to infer properties of the prepared state~\cite{vershynin2018high}.

% The problem
However, in realistic experimental settings, the i.i.d.\ assumption is often violated \cite{schwonnek2021robust, zhang2022device}. The state $\rho_t$ produced in round $t$ may drift in time, depend on uncontrolled latent variables, or respond to measurement settings and outcomes from earlier rounds. Such history dependence arises in systems with slowly evolving environments \cite{white2022nonmarkovian}, in experiments with real-time feedback \cite{wiseman2009quantum}, and in sequential generation schemes whose outputs are correlated \cite{pichler2017universal}. Even without explicit feedback or an engineered memory, gradual miscalibration and continuous phase, polarization, or wavelength drift can produce non-i.i.d.\ data \cite{klimov2018fluctuations}. Therefore, applying statistical tools that assume independent trials can understate uncertainty and overstate device performance \cite{van2013quantum, bavaresco2018measurements, Mannalath2025Sharp}.

Recent work has developed several approaches to quantum characterization beyond the i.i.d.\ setting, with different scopes and statistical costs. Randomized local de Finetti theorems transfer broad classes of quantum-learning protocols from i.i.d.\ inputs to arbitrary quantum states, generally at the cost of additional samples \cite{fawzi2024learning}. For tomography, martingale methods avoid this overhead: projected least-squares tomography and classical shadows retain their i.i.d.\ scaling under fully adaptive preparation when the target is the trajectory-averaged state \cite{zambrano2026quantum, zambrano2026classical}. Such reconstruction can nevertheless be unnecessarily costly when only one observable chosen in advance is needed. For that setting, martingale methods have also been used to construct confidence intervals for average entanglement-witness values under correlated noise, with extensions to general observables \cite{dirkse2020witnessing}. Device-independent approaches have likewise addressed finite-statistics non-i.i.d. verification, first for Bell states \cite{Bancal2021selftestingfinite} and subsequently through a general framework based on self-testing \cite{gocanin2022sample}.

Nevertheless, certification poses a distinct problem: measurements are performed only on a subset of systems, while the guarantee must concern the complementary, remaining systems. Existing spot-checking protocols have relaxed the identical-distribution assumption while retaining independence between rounds \cite{neven2021symmetry, gocanin2022sample, antesberger2026experimental}. Quantum estimation factors provide a powerful, asymptotically tight treatment of sequential non-i.i.d.\ spot-checking \cite{knill2020generation,zhang2026efficient}. However, their construction targets a single prespecified quantity and does not by itself provide a measurement protocol for reusing the same data across many observables. Our framework instead directly adapts existing randomized-measurement protocols to non-i.i.d.\ certification. Protocols for direct fidelity estimation, energy estimation, and entanglement witnessing can then be used with only minor modifications, yielding confidence intervals for their average values over the unmeasured states. Moreover, combined with classical shadows, the framework can efficiently certify many observables simultaneously from the same data.

The statistical core is a martingale description of the sequential experiment \cite{williams1991probability, shah2024introductionmartingales, raginsky2015concentration}. The state prepared in each round is determined by the preceding history, while the test decision and measurement choice are sampled independently of that history. Conditional unbiasedness then ensures that the difference between each single-shot estimator and the corresponding expectation value forms a martingale difference. For verification, Azuma--Hoeffding concentration recovers the standard bounded-estimator sample complexity \cite{fawzi2026optimal}. For spot-checking certification, Freedman's inequality yields a confidence radius whose leading term scales as $1/\sqrt{K}$, where $K$ is the expected number of test rounds. Thus, precision is governed by the systems actually tested, even though the certified statement applies to the remaining systems. We apply this construction to importance-sampled Pauli measurements and classical shadows.

We illustrate the framework in three settings. Numerical simulations demonstrate energy estimation, simultaneous certification of multiple observables, and an entanglement witness under history-dependent coherent drift, and show that a static initial-block strategy can fail to track later use states. We then certify Bell states on 64 qubit pairs of an IBM quantum processor while testing only a fraction of the rounds. These examples cover regimes in which drift, memory, cross-talk, or other temporal correlations make an i.i.d.\ description difficult to justify.

The paper is organized as follows. Section~\ref{sec: abstract protocol} introduces the non-i.i.d.\ source model and abstract estimator. Section~\ref{sec:complexity} derives confidence bounds for verification and spot-checking certification. Section~\ref{sec:application} applies the framework to importance-sampled Pauli measurements and classical shadows. Sections~\ref{sec:numerical} and~\ref{sec:exp} present the numerical and experimental demonstrations, and finally conclude in Section~\ref{sec:conclusion}.

\section{Non-i.i.d. source model and abstract estimator}  \label{sec: abstract protocol}

We consider an adaptive or adversarial experimental scenario in which a source emits a sequence of $N$ quantum states, denoted by $\rho_1, \rho_2, \dots, \rho_N$, each acting on a finite $d$-dimensional Hilbert space. We abandon the standard assumption that these states are independent and identically distributed (i.i.d.). Instead, the state $\rho_t$ prepared at step $t$ may depend arbitrarily on the prior experimental history.

To formalize the non-i.i.d.\ framework \cite{williams1991probability, shah2024introductionmartingales}, we define a probability space $(\Omega, \mathcal{F}, \Pr)$, whose elements specify all possible classical experimental histories, including past preparations, measurement choices and outcomes, as well as uncontrolled or unknown variables (for example, temperature or environmental degrees of freedom). We introduce a filtration $\{\mathcal{F}_t\}_{t=0}^N$, which is a sequence of sub-$\sigma$-algebras of $\mathcal{F}$ satisfying $\mathcal{F}_0 \subseteq \mathcal{F}_1 \subseteq \dots \subseteq \mathcal{F}_N \equiv \mathcal{F}$. Operationally, $\mathcal{F}_t$ contains the complete classical experimental history up to the end of round $t$, including recorded measurement choices and outcomes and any additional classical side information available to the source, environment, or adversary. Mathematically, $\rho_t$ is $\mathcal{F}_{t-1}$-measurable: conditional on the full preceding history, the state at round $t$ is fixed, although it need not be known to the experimenter. Physically, this means that, immediately before the measurement at round $t$, the system is described by a well-defined quantum state determined by the realized history.

A natural realization of this history dependence arises in scenarios where the sequential systems share global entanglement. For instance, the experiment may consist of sequential local measurements performed on the subsystems of a single, globally entangled multipartite state $\rho^{(N)}$. After measuring the first $t-1$ subsystems, the effective state of the $t$-th subsystem is given by the conditional reduced density matrix
\begin{align}
    \rho_{t}
    =
    \frac{
    \tr_{\setminus t} \left[
    \left(E_{k_1} \otimes \dots \otimes E_{k_{t-1}} \otimes \mathds{1}_{t \dots N}\right) \rho^{(N)}
    \right]
    }{
    \Pr(k_1,\dots,k_{t-1})
    },
\end{align}
where $E_{k_i}$ are the POVM elements associated with the realized classical outcomes $k_1,\dots,k_{t-1}$, $\Pr(k_1,\dots,k_{t-1})$ is the probability of the observed history, and $\tr_{\setminus t}$ denotes the partial trace over all subsystems except the $t$-th.

Our objective is to estimate the expectation value of a Hermitian observable $O=\alpha_I\mathds{1}+W$, where $\tr(W)=0$ and $\alpha_I \in \mathds{R}$. Since the identity part is deterministic, the only statistical task is to estimate
\begin{align}
    \omega_t=\tr(W\rho_t).
\end{align}
Then, the full expectation value is $\tr(O\rho_t) = \alpha_I + \omega_t$.

To estimate $\omega_t$ on a measured round, the experimenter relies on a trusted, randomized measurement procedure. The experimenter privately selects a measurement setting independently of the preceding experimental history $\mathcal{F}_{t-1}$ and records the resulting quantum outcome. This joint process of randomly choosing a setting and performing the corresponding measurement can be described by a single effective POVM $\{E_k\}_{k=1}^M$, where the index $k$ encompasses both the classical setting choice and the physical outcome. While an adversary may know the overall POVM description and tailor $\rho_t$ accordingly, she cannot predict the measurement setting choice or influence the physical apparatus. Because $\rho_t$ is fixed conditional on $\mathcal{F}_{t-1}$, and the setting choice is independent of $\mathcal{F}_{t-1}$, the outcome $k$ at round $t$ has conditional probability
\begin{align}
\Pr(k \mid \mathcal{F}_{t-1}) = \tr(\rho_t E_{k}).
\end{align}

Based on the observed outcome $k$, the protocol defines a single-shot estimator $X(k)$. We require this estimator to satisfy two properties: first, for any  quantum state $\rho$, the estimator must be unbiased, meaning that its expectation over the POVM outcome equals the target value
\begin{align}
    \mathbb{E}_{k\sim p_\rho}[X(k)]
    =
    \sum_{k=1}^M \tr(\rho E_k)\,X(k)
    =
    \tr(W\rho),
\end{align}
where $p_\rho(k)=\tr(\rho E_k)$. When the protocol is applied to the state $\rho_t$ generated at round $t$, this property becomes conditional unbiasedness. Namely, if $X_t=X(k_t)$ denotes the estimator associated with the realized outcome $k_t$, then
\begin{align}\label{eq:exp_val_Xt_past}
    \mathbb{E}[X_t\mid\mathcal{F}_{t-1}]
    =
    \sum_{k=1}^M
    \Pr(k_t=k\mid\mathcal{F}_{t-1})X(k)
    =
    \omega_t.
\end{align}

The second property to be satisfied is that the single-shot estimator must lie in a known bounded interval $[a,b]$, so that $X_t\in[a,b]$ on every measured round.

These two properties already suffice for Azuma--Hoeffding concentration. To derive variance-sensitive bounds, we also characterize the estimator's conditional second moment using the operator
\begin{align}\label{eq:second-moment-operator}
    Q
    =
    \sum_{k=1}^M X(k)^2E_k.
\end{align}
 Thus, for every state and every experimental history,
\begin{align}
    \mathbb{E}[X_t^2 \! \mid \!  \mathcal F_{t-1}]
    =
    \tr(\rho_tQ_{}) \leq \!\!  \sup_{\substack{\sigma\geq0, \\ \, \tr(\sigma)=1}} \!\!\!\!\tr(\sigma Q_{})
    =
    \norm{Q_{}}_\infty,
\end{align}
where $\norm{\cdot}_\infty$ is the operator norm. 

Depending on the task, the experimenter may choose to measure all generated states, corresponding to verification, or only a randomly selected subset of them, corresponding to certification.

\section{Confidence bounds} \label{sec:complexity}

To derive confidence intervals without assuming independence, we rely on the framework of martingales. Formally, a sequence of real-valued random variables $\{Z_t\}_{t=1}^N$ forms a martingale difference sequence with respect to a filtration $\{\mathcal{F}_t\}_{t=0}^N$ if each $Z_t$ is $\mathcal{F}_t$-measurable, $\mathbb{E}[|Z_t|] < \infty$, and
\begin{align}
\mathbb{E}[Z_t \mid \mathcal{F}_{t-1}] = 0 \quad \text{for all } t \ge 1.
\end{align}
If we define the deviation of our single-shot estimator from the round-specific expectation value as $Z_t = X_t - \omega_t$, then this sequence forms a martingale difference sequence due to Eq.~\eqref{eq:exp_val_Xt_past}.

The sequence of partial sums $M_n = \sum_{t=1}^n Z_t$, with $M_0=0$, then forms a martingale. Martingales allow us to bound the accumulation of random fluctuations using concentration inequalities.

\begin{theorem}[Azuma--Hoeffding~\cite{azuma1967weighted, raginsky2015concentration}] \label{thm:azuma}
Let $\{M_t\}_{t=1}^N$ be a martingale with respect to a filtration $\{\mathcal{F}_t\}_{t=0}^N$, with $M_0 = 0$, and write $ Z_t=M_t-M_{t-1}$. If there exist predictable processes $\{A_t\}_{t=1}^N$ and $\{B_t\}_{t=1}^N$, and constants $0<c_1, \dots, c_N <\infty$, such that
\begin{align}
    A_t \le Z_t \le B_t
    \quad \text{and} \quad
    B_t - A_t \leq c_t
\end{align}
almost surely for all $t=1,\dots,N$, then for any $\epsilon > 0$,
\begin{align}
\Pr (|M_N| \ge \epsilon)
\le
2 \exp\left( -\frac{2\epsilon^2}{\sum_{t=1}^N c_t^2} \right).
\end{align}
\end{theorem}

\begin{theorem}[Freedman~\cite{freedman1975tail, raginsky2015concentration,
shah2024introductionmartingales}]\label{thm:freedman}
Let $\{M_t\}_{t=0}^N$ be a martingale with respect to a filtration $\{\mathcal{F}_t\}_{t=0}^N$, with $M_0=0$, and write $ Z_t=M_t-M_{t-1}$. Suppose that $|Z_t|\leq R$ almost surely for all $t=1,\dots,N$. Define the predictable quadratic
variation as
\begin{align}
    \mathcal{V}_N
    =
    \sum_{t=1}^N
    \mathbb{E}[Z_t^2\mid\mathcal{F}_{t-1}].
\end{align}
Then, for every $\epsilon, v>0$,
\begin{align}\label{eq:freedman-tail}
    \Pr\!\left(|M_N|\geq \epsilon \text{ and }\mathcal{V}_N\leq v \right) \leq 2\exp\!\left[-\frac{\epsilon^2/2}{v+R\epsilon/3}\right].
\end{align}
Consequently, if $\mathcal{V}_N\leq v$ almost surely, then, for every
$\delta\in(0,1)$, with probability at least $1-\delta$,
\begin{align}\label{eq:freedman-inverted}
    |M_N| \leq \frac{R}{3}\ln\!\left(\frac{2}{\delta}\right) + \sqrt{2v\ln\!\left(\frac{2}{\delta}\right)+\left[\frac{R}{3}\ln\!\left(\frac{2}{\delta}\right)\right]^2}.
\end{align}
\end{theorem}

We use these inequalities in two different settings. For verification, where every state is measured, Azuma--Hoeffding recovers the familiar behavior of standard i.i.d.\ analyses. Certification requires more care because only a randomly selected fraction of the states is measured, and the information obtained from those states must be appropriately reweighted. In this sparse-testing setting, a bound based only on the largest possible single-round contribution can be overly conservative. Freedman's inequality also incorporates second-moment information, allowing it to account for the sparsity of the test rounds and yield tighter confidence bounds, particularly when the testing probability is small.

\subsection{Quantum Verification} \label{sec: verification protocol}

In the verification protocol, the experimenter consumes all $N$ states $\rho_1, \dots, \rho_N$ generated by the source to estimate the trajectory-averaged observable value
\begin{align}
    \bar{\omega} = \frac{1}{N} \sum_{t=1}^N \omega_t.
\end{align}
To achieve this, the protocol computes the empirical average of the recorded single-shot estimators introduced in Section~\ref{sec: abstract protocol}
\begin{align}
    \bar{X} = \frac{1}{N} \sum_{t=1}^N X_t.
\end{align}

\begin{theorem}[Verification] \label{thm:verification}
For single-shot estimators $X_t \in [a, b]$, the empirical average $\bar{X}$ satisfies
\begin{align}
    |\bar{X} - \bar{\omega}| \leq \epsilon
\end{align}
with probability at least $1-\delta$, provided that
\begin{align}\label{eq:sample complexity verification}
    N \geq \frac{(b-a)^2}{2 \epsilon^2}\ln\left(\frac{2}{\delta}\right).
\end{align}
\end{theorem}

\begin{proof}
We define the difference sequence $Z_t$ as the deviation of the single-shot estimator from the round-specific value at step $t$, such that
\begin{align}
    Z_t = X_t - \omega_t.
\end{align}
The conditional expectation vanishes, $\mathbb{E}[Z_t \mid \mathcal{F}_{t-1}] = 0$. Consequently, the sequence $M_n = \sum_{t=1}^n Z_t$, with $M_0 = 0$, forms a martingale.

Since $X_t \in [a,b]$ and the round-specific expectation value $\omega_t$ is $\mathcal{F}_{t-1}$-measurable, the martingale difference is bounded as
\begin{align}
A_t = a - \omega_t \le Z_t \le b - \omega_t = B_t.
\end{align}
Crucially, the $\omega_t$ dependence cancels, yielding
\begin{align}
B_t - A_t = b - a = c_t.
\end{align}
Applying Theorem~\ref{thm:azuma} to $M_N$, we bound the probability that the empirical average $\bar{X} = \frac{1}{N}\sum_{t=1}^N X_t$ deviates from the true average $\bar{\omega}$ by more than $\epsilon$,
\begin{align}
    \Pr\left( \left| \bar{X} - \bar{\omega} \right| \ge \epsilon \right)
    &= \Pr\left( |M_N| \ge N\epsilon \right) \nonumber \\
    &\le 2 \exp\left( -\frac{2 N\epsilon^2}{(b-a)^2} \right).
\end{align}

Requiring this upper bound to be at most $\delta$ gives
\begin{align}
N \geq \frac{(b-a)^2}{2 \epsilon^2}\ln\left(\frac{2}{\delta}\right).
\end{align}
Therefore, this number of measurements is sufficient to estimate $\bar{\omega}$ up to additive error $\epsilon$ with probability at least $1-\delta$.
\end{proof}

\subsection{Quantum Certification}

In many protocols, state verification is not the primary task, but a subroutine used to ensure that the prepared states possess the properties required by a parallel, intended quantum task. Because quantum measurements are generally destructive, full verification consumes the very resource it is meant to guarantee. To retain a usable output, the experimenter must instead test a randomly selected subset of the states, using those outcomes to rigorously certify the average observable value of the remaining systems \cite{gocanin2022sample, antesberger2026experimental, knill2020generation, zhang2026efficient}.

At each step $t$, prior to interacting with the state $\rho_t$, the experimenter generates a private classical coin flip $c_t \in \{0, 1\}$ with $\Pr(c_t = 1) = p$. Crucially, this choice is independent of the preceding experimental history $\mathcal{F}_{t-1}$, and hence independent of the prepared state $\rho_t$. According to the value of $c_t$, the protocol proceeds in one of two ways:
\begin{enumerate}
    \item[(i)] \textit{Test round ($c_t = 1$):} The state $\rho_t$ is measured using the trusted effective POVM $\{E_k\}_{k=1}^M$ introduced in Section~\ref{sec: abstract protocol}. The experimenter records the outcome $k_t$ and computes the single-shot estimator $X_t=X(k_t)$ using the unbiased estimator $X$.
    \item[(ii)] \textit{Use round ($c_t = 0$):} The corresponding system is passed directly to the intended quantum task. No estimate is generated for this round.
\end{enumerate}

Let $S_{T} = \{t  \mid c_t = 1\}$ and $S_U = \{t  \mid c_t = 0\}$ denote the disjoint, random sets of indices corresponding to the test and use rounds, respectively, satisfying $|S_T| + |S_U| = N$. Because $\rho_t$ is completely determined by $\mathcal{F}_{t-1}$, its round-specific expectation value $\omega_t = \tr(W \rho_t)$ is a well-defined quantity regardless of the experimenter's subsequent choice $c_t$.

Our objective is to bound the sum of the expected values over the use-round states, $\sum_{t \in S_U} \omega_t$, using only the classical data $\{X_t\}_{t \in S_T}$ collected during the test rounds.

\begin{theorem}[Certification]\label{thm:certification}
Let $W$ be a Hermitian observable, and suppose that $X_t\in[a,b]$
on every tested round. Set $q=1-p$ and choose $R$ and $v$ satisfying
\begin{align}
    R
    &\geq
    \max\left\{
        \norm{W}_\infty,
        \frac{q}{p}\max\{|a|,|b|\}
    \right\},\\
    v
    &\geq
    \norm{qW^2+\frac{q^2}{p}Q_{}}_\infty.
\end{align}
Then, if $|S_U|>0$, the following holds with probability at least $1-\delta$,
\begin{align} \label{eq: certification bound thrm 3}
 &\left| \frac{1}{|S_U|}\sum_{t \in S_U} \omega_t - \frac{q}{p|S_U|} \sum_{t \in S_T} X_t \right| \nonumber \\
 &\qquad \leq \frac{\frac{R}{3}\ln\!\left(\frac{2}{\delta}\right) + \sqrt{ 2Nv\ln\!\left(\frac{2}{\delta}\right) +\left[\frac{R}{3}\ln\!\left(\frac{2}{\delta}\right)\right]^2}}{|S_U|}.
\end{align}
\end{theorem}

Crucially, each system produced on a use round must be passed directly to the intended quantum task in the same round. Theorem~\ref{thm:certification} bounds the sum of their observable values $\sum_{t \in S_U} \omega_t$ at the moment the states are generated. If these states are stored for later use, they may become vulnerable to subsequent decoherence or to adversarial noise once the test locations are revealed.

\begin{proof}

To relate the tested data to the unmeasured states, we introduce the random variable \cite{horvitz1952generalization, knill2020generation,  zhang2026efficient}
\begin{align}
V_t = (1 - c_t) \omega_t - \frac{q}{p} c_t X_t.
\end{align}
Here, $X_t$ denotes the single-shot estimator associated with the test measurement at round $t$, introduced in Section~\ref{sec: abstract protocol}. It is physically observed only when $c_t=1$. Since it appears multiplied by $c_t$, the quantity $c_tX_t$ is well defined from the observed data, with the convention that $c_t X_t=0$ on use rounds.

We evaluate the expectation of $V_t$ conditioned on $\mathcal{F}_{t-1}$ by averaging over both the private coin $c_t$ and the measurement outcome on test rounds. Thus,
\begin{align}
\mathbb{E}[V_t |\mathcal{F}_{t-1}] &=
\mathbb{E}\left[(1-c_t)\omega_t \mid \mathcal{F}_{t-1}\right]
- \frac{q}{p} \mathbb{E}\left[c_t X_t \mid \mathcal{F}_{t-1}\right]
\nonumber \\
&= q\omega_t - \frac{q}{p} \, p\, \mathbb{E}\left[X_t \mid \mathcal{F}_{t-1}, c_t=1\right] \nonumber \\
&= q\omega_t - \frac{q}{p}p\,\omega_t =0.
\end{align}
The third line follows because the coin flip is independent of $\mathcal{F}_{t-1}$, and hence of $\rho_t$. Moreover, on the event $c_t=1$, the test measurement is conditionally unbiased, so that $\mathbb{E}[X_t \mid \mathcal{F}_{t-1}, c_t=1] = \omega_t$.

Thus, the sequence of partial sums $M_n = \sum_{t=1}^n V_t$, with $M_0=0$, forms a martingale. By definition, $M_N$ compares the sum of the observable values over the remaining systems with the scaled sum of the test-round estimators:
\begin{align}
\label{eq:M_N_certification}
M_N = \sum_{t \in S_U} \omega_t - \frac{q}{p} \sum_{t \in S_T} X_t.
\end{align}

We now bound both the magnitude and the conditional second moment of the martingale difference. If $c_t=0$, then $V_t=\omega_t$, and hence
\begin{align}
    |V_t|
    =
    |\tr(W\rho_t)|
    \leq
    \norm{W}_\infty.
\end{align}
If $c_t=1$, then $V_t=-qX_t/p$, and hence
\begin{align}
    |V_t|
    \leq
    \frac{q}{p}\max\{|a|,|b|\}.
\end{align}
Therefore, $|V_t|\leq \max \{
        \norm{W}_\infty,
        \frac{q}{p}\max\{|a|,|b|\}
    \}\leq R$. 
Furthermore, averaging over the private coin and, on a test round, over the measurement outcome gives
\begin{align}
\mathbb{E}[V_t^2\mid\mathcal{F}_{t-1}] &= q\,\omega_t^2 + p\frac{q^2}{p^2}
\mathbb{E}[X_t^2\mid\mathcal{F}_{t-1},c_t=1] \nonumber\\
&\leq \tr\!\left[\rho_t\left(qW^2+\frac{q^2}{p}Q_{} \right) \right] \leq v.
\end{align}
Then, the predictable quadratic variation satisfies
\begin{align}
    \mathcal{V}_N
    =
    \sum_{t=1}^N
    \mathbb{E}[V_t^2\mid\mathcal{F}_{t-1}]
    \leq Nv.
\end{align}

Applying Theorem~\ref{thm:freedman} with bound $R$ and quadratic-variation bound $Nv$ yields, with probability at least $1-\delta$,
\begin{align}
\label{eq:bound_certification_no_normalized}
&\left|\sum_{t \in S_U} \omega_t- \frac{q}{p} \sum_{t \in S_T} X_t
\right| \nonumber \\
& \quad\leq \frac{R}{3}\ln\!\left(\frac{2}{\delta}\right) + \sqrt{ 2Nv\ln\!\left(\frac{2}{\delta}\right) + \left[\frac{R}{3}\ln\!\left(\frac{2}{\delta}\right)\right]^2}.
\end{align}
Provided that at least one state is used in the parallel protocol, division by $|S_U|$ gives Eq.~\eqref{eq: certification bound thrm 3}.

\end{proof}

 The confidence interval in Theorem~\ref{thm:certification} uses the
observed number of use rounds. The following corollary gives a sufficient
number of rounds for a prescribed accuracy.

\begin{corollary}
\label{cor:finite-certification-budget}
Under the assumptions of Theorem~\ref{thm:certification}, fix
$\epsilon>0$, $\delta\in(0,1)$, and $N$ satisfying both
\begin{align}
    N\geq\frac{8}{q}\ln\!\left(\frac{2}{\delta}\right), \quad 
    N\geq
    \left(\frac{8v}{q^2\epsilon^2}
    +\frac{4R}{3q\epsilon}\right)
    \ln\!\left(\frac{4}{\delta}\right).
    \label{eq:finite-budget-use-count}
\end{align}
Then, considering $|S_U|>0$, with probability at least $1-\delta$ we have
\begin{align}
    \left|
    \frac{1}{|S_U|}\sum_{t\in S_U}\omega_t
    -\frac{q}{p|S_U|}\sum_{t\in S_T}X_t
    \right|\leq\epsilon.
\end{align}
\end{corollary}

\begin{proof}
Let $M_N$ be as in Eq.~\eqref{eq:M_N_certification}. Since the private coin is chosen in an i.i.d. manner, $|S_U|$ follows a binomial distribution with parameters $N$ and $q$. Thus, a Chernoff bound and Eq.~\eqref{eq:finite-budget-use-count} give
\begin{align}
    \Pr(|S_U|<qN/2)\leq e^{-qN/8}\leq\delta/2.
\end{align}
Applying Eq.~\eqref{eq:bound_certification_no_normalized}
with $\delta$ replaced by $\delta/2$ gives, with probability at least $1-\delta/2$,
\begin{align}
    \label{eq:M_N_2_delta}
    |M_N| \leq
    \frac{R}{3}\ln(\frac{4}{\delta})
    +\sqrt{2Nv\ln(\frac{4}{\delta})+\left[\frac{R }{3}\ln(\frac{4}{\delta})\right]^2}.
\end{align}
The second inequality in Eq.~\eqref{eq:finite-budget-use-count} ensures $x = \epsilon qN/2 \geq2R \ln({4}/{\delta})/3$ and
\begin{align}
    2Nv\ln(\frac{4}{\delta})+\left[\frac{R}{3}\ln(\frac{4}{\delta})\right]^2
    \leq
    \left(x-\frac{R}{3}\ln(\frac{4}{\delta})\right)^2.
\end{align}
Substituting in Eq.~\eqref{eq:M_N_2_delta} we obtain $|M_N| \leq x$.

By a union bound, with probability at least $1-\delta$,
$|S_U|\geq qN/2>0$ and $|M_N|\leq\epsilon qN/2$.
Consequently, $|M_N|/|S_U|\leq\epsilon$.
\end{proof}

These sufficient design conditions are conservative.
The observed-count interval in Theorem~\ref{thm:certification}
remains the appropriate expression for reporting the uncertainty
of a completed run.

\section{Observable estimation protocols}\label{sec:application}

In this section, we specialize the general verification and certification bounds to two common strategies for estimating quantum observables. We first consider importance-sampled Pauli measurements, which are tailored to a single target observable, and then classical shadows, which allow the same spot-checking data to certify a family of observables simultaneously.

\subsection{Importance-sampled Pauli estimation}

Importance-sampled Pauli measurements are standard in Hamiltonian and
variational estimation \cite{wecker2015progress,crawford2021efficient},
with optimality results known for fidelity estimation
\cite{fawzi2026optimal}. We now specialize our framework to this
measurement strategy for estimating a fixed Hermitian observable $O$
under non-i.i.d.\ state preparation.

An $n$-qubit Hermitian operator $O$ can be decomposed in the Pauli basis
$\mathcal{P}_n=\{I,X,Y,Z\}^{\otimes n}$ as
\begin{align}
    O
    =
    \sum_{P\in\mathcal{P}_n}\alpha_P P,
    \qquad
    \alpha_P
    =
    \frac{\tr(OP)}{d}.
    \label{eq:O generic}
\end{align}
The identity term $\alpha_I\mathds{1}$ contributes deterministically to
$\tr(O\rho_t)$ and is therefore added analytically. The randomized
measurement protocol is applied only to the traceless part
\begin{align} \label{eq:traceless operator W}
    W
    =
    \sum_{P\in\mathcal{P}_n\setminus\{I\}}\alpha_P P.
\end{align}
For $W\neq0$, define the importance sampling distribution
\begin{align}
    \pi(P)
    =
    \frac{|\alpha_P|}{\|\alpha\|_1},
    \qquad
    \|\alpha\|_1
    =
    \sum_{P\in\mathcal{P}_n\setminus\{I\}}|\alpha_P|.
\end{align}
On a measured round $t$, the experimenter privately samples a Pauli observable $P_t$ from the distribution $\pi(P)$ and measures it. The outcome $y_t \in \{+1,-1\}$ allows us to compute
    \begin{align}
        \label{eq:single_shot_estimator_Pauli}
        X_t
        =
        \|\alpha\|_1\,y_t\,
        \operatorname{sgn}(\alpha_{P_t}).
    \end{align}

The random choice of Pauli observable and the corresponding measurement outcome on a measured round can be described jointly by the POVM
\begin{align}
    E_y^P
    =
    \frac{\pi(P)}{2}
    \left(\mathds{1}+yP\right).
\end{align}
Then, the single-shot estimator in Eq.~\eqref{eq:single_shot_estimator_Pauli} satisfies
\begin{align}
    \mathbb{E}[X_t\mid\mathcal{F}_{t-1}, c_t=1]
    =
    \omega_t
    =
    \tr(W\rho_t)
\end{align}
and is bounded as $|X_t|\leq\|\alpha\|_1$.
The associated second-moment operator is
\begin{align}
    Q_{}
    &=
    \sum_{P,y}X(P,y)^2E_y^P
    =
    \|\alpha\|_1^2\mathds{1}.
\end{align}
This allows us to use the results in Section~\ref{sec:complexity} to
estimate the average observable value over a set of rounds,
\begin{align}
    \bar{\omega}_S
    =
    \frac{1}{|S|}
    \sum_{t\in S}\omega_t.
\end{align}

In the verification setting, every state is measured, so that $S=\{1,\ldots,N\}$. Applying Theorem~\ref{thm:verification} with interval width $b-a=2\|\alpha\|_1$, we obtain that, to estimate $\bar{\omega}_{S_N}$ up to additive error $\epsilon$ with probability at least $1-\delta$, it is sufficient to take
\begin{align}
    N
    \geq
    \frac{2\|\alpha\|_1^2}{\epsilon^2}
    \ln\!\left(\frac{2}{\delta}\right).
\end{align}

On the other hand, the spot-checking certification protocol based on
importance-sampled Pauli measurements can be implemented as follows:
\begin{enumerate}
    \item \textit{Measurement design}. Fix the target observable
    $O=\alpha_I\mathds{1}+W$, construct the Pauli sampling distribution
    $\pi(P)$ from Eq.~\eqref{eq:O generic}, set $q=1-p$ and choose the number of
    rounds $N$, testing probability $p\in(0,1)$, and failure
    probability $\delta$. 
    \item \textit{Data collection}. For each round $t$, privately sample
    $c_t\in\{0,1\}$, with $\Pr(c_t=1)=p$, independently of the
    preceding experimental history. 
    \begin{itemize}
        \item If $c_t=1$, privately sample a
    Pauli observable $P_t$ according to $\pi(P)$, measure $P_t$,
    record the outcome $y_t$, and compute $X_t$ as in Eq. \eqref{eq:single_shot_estimator_Pauli}.
    \item If $c_t=0$, pass the system directly to the intended quantum task.
    \end{itemize}
    \item \textit{Certification}. Let $S_T$ and $S_U$ denote the
    observed sets of test and use rounds. If $|S_U|>0$, output
    \begin{align}
        \hat{o}_U
        =
        \alpha_I
        +
        \frac{q}{p|S_U|}
        \sum_{t\in S_T}X_t.
    \end{align}
\end{enumerate}

The bounds above allow us to apply Theorem~\ref{thm:certification} with
\begin{align}
    R
    &=
    \max\left\{
        \norm{W}_\infty,
        \frac{q}{p}\|\alpha\|_1
    \right\},
    \label{eq:pauli-R-v}\\
    v
    &=
    q\norm{W}_\infty^2
    +
    \frac{q^2}{p}\|\alpha\|_1^2.
\end{align}
The resulting confidence radius is
\begin{align}
    \epsilon
    =
    \frac{
        \frac{R}{3}\ln(2/\delta)
        +
        \sqrt{
            2Nv\ln(2/\delta)
            +
            \left[
                \frac{R}{3}\ln(2/\delta)
            \right]^2
        }
    }{|S_U|}.
    \label{eq:pauli-freedman-radius}
\end{align}
With probability at least $1-\delta$, 
the average observable value of the unmeasured states,
$\bar{o}_U = \alpha_I + \bar{\omega}_{S_U}$, satisfies
$|\hat{o}_U-\bar{o}_U|\leq\epsilon$.

To make the scaling explicit, suppose that $p\leq1/2$. Since
$\norm{W}_\infty\leq\|\alpha\|_1$ and $q/p\geq1$, the parameters
above satisfy
\begin{align}
    R
    &=
    \frac{q}{p}\|\alpha\|_1,
    \qquad
    v
    \leq
    \frac{q}{p}\|\alpha\|_1^2.
\end{align}
Let $K=pN$ denote the expected number of test rounds.
For a target accuracy $0<\epsilon\leq\|\alpha\|_1$,
Corollary~\ref{cor:finite-certification-budget} guarantees
$|\hat{o}_U-\bar{o}_U|\leq\epsilon$ with probability at least
$1-\delta$ whenever
\begin{align}
    K
    \geq
    \left(
        \frac{8\|\alpha\|_1^2}{q\epsilon^2}
        +\frac{4\|\alpha\|_1}{3\epsilon}
    \right)
    \ln\!\left(\frac{4}{\delta}\right).
    \label{eq:sparse-K-scaling}
\end{align}
Under the assumed range of $\epsilon$, this bound also implies
the use-round count condition in the first inequality of
Eq.~\eqref{eq:finite-budget-use-count}.
Thus, in the sparse-testing regime $q\simeq1$, the certification
accuracy is governed by the expected number $K$ of test rounds.

\subsection{Spot-checking classical shadows}\label{sec:spot-shadows}

The preceding protocol is tailored to a single observable fixed before data collection. Classical shadows instead allow the same measurement data to estimate a family of observables simultaneously \cite{huang2020predicting,zambrano2026classical}. We extend this capability to spot-checking certification: measurements on randomly selected test rounds certify the average value of every observable in the family over the complementary, remaining systems. Following Ref.~\cite{zambrano2026classical}, we obtain finite-sample guarantees by combining truncated shadow estimators, controlled through the shadow norm, with the spot-checking construction developed above.

Let $\{E_k\}_{k=1}^M$ be the informationally complete (IC) POVM describing the randomized shadow measurement, where $k$ includes both the measurement setting and its outcome. For $k$, the classical-shadow construction assigns a Hermitian single-shot estimator $\hat{\rho}_k$ whose average reconstructs the measured state
\begin{align}\label{eq:shadow-reconstruction}
    \mathbb{E}_{k\sim p_\rho}[\hat{\rho}_k] = \sum_{k=1}^M \tr(\rho E_k)\hat{\rho}_k = \rho.
\end{align}

For a fixed traceless observable $W$, define the corresponding single-shot estimator $X(k)=\tr(W\hat{\rho}_k)$. Eq.~\eqref{eq:shadow-reconstruction} immediately gives
\begin{align}
    \mathbb{E}_{k\sim p_\rho}[X(k)]
    =
    \tr(W\rho).
\end{align}
Thus, $X(k)$ is an unbiased estimator of the expectation value of $W$. For this estimator, the worst-case second moment introduced in
Eq.~\eqref{eq:second-moment-operator} is the shadow norm squared
\cite{huang2020predicting,zambrano2026classical}:
\begin{align}\label{eq:shadow-norm}
    \norm{W}_{\rm sh}^2 = \sup_{\substack{\rho \geq 0 \\ \tr(\rho) = 1}}\sum_{k=1}^M \tr(\rho E_k) \left[\tr(W\hat{\rho}_k)\right]^2 .
\end{align}

To control large single-shot estimates, we introduce a threshold $T>0$ and define
\begin{align}\label{eq:truncated-shadow-score}
    Y(k)
    =
    \begin{cases}
        X(k), & |X(k)|\leq T,\\
        T\,\operatorname{sgn}(X(k)), & |X(k)|>T.
    \end{cases}
\end{align}
The truncated estimator satisfies $|Y(k)|\leq T$, at the cost of a bias controlled by the shadow norm. This bias is incorporated into the confidence bound below.

The spot-checking classical-shadow protocol is:
\begin{enumerate}
    \item \textit{Measurement design}. Fix an IC measurement $\{E_k\}_{k=1}^M$ and its reconstruction rule $k\mapsto\hat{\rho}_k$. Choose $N$, $p\in(0,1)$, $\delta\in(0,1)$, and set $q=1-p$.
    \item \textit{Data collection}. For each round $t$, privately sample $c_t\in\{0,1\}$, with $\Pr(c_t=1)=p$, independently of the preceding experimental history. 
    \begin{itemize}
        \item If $c_t=1$, perform the IC measurement using private randomness and record the outcome $k_t$.
        \item If $c_t=0$, pass the system directly to the intended quantum task.
    \end{itemize} 
    \item \textit{Shadow estimation and certification}. Choose a family
    $O_j=\alpha_{j}\mathds{1}+W_j$ and a common truncation threshold $T>0$, independently of the observed data. For every $j=1,\ldots,L$ and every test round $t\in S_T$, compute
    \begin{align}
        X_{t,j}
        =
        \tr(W_j\hat{\rho}_{k_t})
    \end{align}
    and truncate it at $T$ according to Eq.~\eqref{eq:truncated-shadow-score}, obtaining $Y_{t,j}$. If
    $|S_U|>0$, output
    \begin{align}
        \hat{o}_{j}
        =
        \alpha_j
        +
        \frac{q}{p|S_U|}
        \sum_{t\in S_T}Y_{t,j}.
    \end{align}
\end{enumerate}

For each observable, the quantity to be certified is its average value
over the use-round states,
\begin{align}
    \bar{o}_{j} =
    \alpha_j
    +
    \frac{1}{|S_U|}
    \sum_{t\in S_U}\tr(W_j\rho_t).
\end{align}
The following theorem specifies a common threshold and provides a
simultaneous confidence bound for all $L$ observables.

\begin{theorem}[Spot-checking classical shadows]
\label{thm:spot-shadows}
Let $\{O_j\}_{j=1}^L$ be a finite family of Hermitian observables chosen
independently of the observed data, and decompose
\begin{align}
    O_j = \alpha_{j}\mathds{1}+W_j,  \qquad \tr(W_j)=0.
\end{align} 
For $\epsilon>0$, define
\begin{align}
    S^2 \!=\! \max_{j}\!\norm{W_j}_{\rm sh}^2, \quad
     T \!=\! \frac{5S^2}{4\epsilon}, \quad
     R \!=\! T\max\!\left(1,\frac{q}{p}\right).
\end{align}

Run the protocol above using the common threshold $T$. Then, with
probability at least $1-\delta$, whenever $|S_U|>0$,
\begin{align}\label{eq:spot-shadow-bound}
    &\max_{1\leq j\leq L}
    \left|
        \bar{o}_{j}-\hat{o}_{j}
    \right|
    \leq
    \frac{\epsilon}{5}
    +
    \frac{R}{3|S_U|}
    \ln\!\left(\frac{2L}{\delta}\right)
    \nonumber\\
    & \qquad \: +
    \frac{1}{|S_U|}
    \sqrt{
        \frac{2Nq}{p}S^2
        \ln\!\left(\frac{2L}{\delta}\right)
        +
        \left[
            \frac{R}{3}
            \ln\!\left(\frac{2L}{\delta}\right)
        \right]^2
    }.
\end{align}
\end{theorem}

\begin{proof}
For each $j$, let $X_j(k)=\tr(W_j\hat{\rho}_k)$ and let $Y_j(k)$ be its truncation at $T$, as defined in Eq.~\eqref{eq:truncated-shadow-score}. Let the effective observable and second-moment operator associated with the truncated estimator be
\begin{align}
    \tilde{W}_j
    &=
    \sum_{k=1}^M Y_j(k)E_k,
    &
    Q_{j}
    &=
    \sum_{k=1}^M Y_j(k)^2E_k.
\end{align}
By construction, $Y_{t,j}=Y_j(k_t)$ is conditionally unbiased for the
effective observable $\tilde W_j$:
\begin{align}
    \mathbb{E} [Y_{t,j}\mid\mathcal{F}_{t-1},c_t=1] = \tr(\tilde W_j\rho_t).
\end{align}

Since $|Y_j(k)|\leq T$ and $\sum E_k =\mathds{1}$, we have  $\Vert \tilde W_j \Vert_\infty \leq T$. Moreover, the operator Jensen inequality gives \cite{hansen2003jensen}
\begin{align}
    \tilde W_j^2
    =
    \left(\sum_{k=1}^M Y_j(k)E_k\right)^2
    \leq 
    \sum_{k=1}^M Y_j(k)^2E_k
    =
    Q_{j}.
\end{align}
Truncation can only decrease the squared estimator, $Q_{j} \leq  \sum_{k=1}^M X_j(k)^2E_k$ and hence $\norm{Q_{j}}_\infty \leq \norm{W_j}_{\rm sh}^2 \leq S^2$. It follows that the second-moment quantity in Theorem~\ref{thm:certification} satisfies
\begin{align}
    \left\Vert q\tilde W_j^2 + \frac{q^2}{p}Q_{j} \right\Vert_\infty
    \leq \left(q+\frac{q^2}{p}\right) \Vert Q_{j} \Vert_\infty \leq \frac{q}{p}S^2.
\end{align}
Since both $\tilde{W}_j$ and $|Y_j (k)|$ are bounded by $T$, the increment bound in Theorem~\ref{thm:certification} can be chosen as $R = T\max\!\left(1,\frac{q}{p}\right)$.

For $|S_U|>0$, define the average value of the observable associated with the truncated estimator by
\begin{align}
    \tilde{o}_{j} = \alpha_j + \frac{1}{|S_U|} \sum_{t\in S_U}\tr(\tilde W_j\rho_t).
\end{align}
Applying Theorem~\ref{thm:certification} to $\tilde W_j$ and $Y_{t,j}$ with confidence parameter $\delta/L$, followed by a union bound over $j=1,\ldots,L$, shows that, with probability at least $1-\delta$, the error $|\tilde{o}_{j}-\hat{o}_{j}|$ is bounded simultaneously for all $j$ by the last two terms in
Eq.~\eqref{eq:spot-shadow-bound}. Here, the identity contributions $\alpha_j$ cancel from the difference.

It remains to control the truncation bias $|\bar{o}_{j}-\tilde{o}_{j}|$. From Eq.~\eqref{eq:shadow-reconstruction}, $W_j=\sum X_j(k)E_k$. Using $(u-T)_+\leq u^2/(4T)$ for $u\geq0$, we obtain, for every state $\rho$,
\begin{align}
    \left|\tr[(W_j-\tilde W_j)\rho]\right|
    &\leq
    \sum_{k=1}^M
    \tr(\rho E_k)
    |X_j(k)-Y_j(k)|
    \nonumber\\
    &\leq
    \frac{1}{4T}
    \sum_{k=1}^M
    \tr(\rho E_k)X_j(k)^2
    \nonumber\\
    &\leq
    \frac{\norm{W_j}_{\rm sh}^2}{4T}
    \leq
    \frac{\epsilon}{5}.
\end{align}
Averaging over the use rounds gives $|\bar{o}_{j}-\tilde{o}_{j}|\leq \frac{\epsilon}{5}$.

Finally, $|\bar{o}_{j}-\hat{o}_{j}| \leq |\bar{o}_{j}-\tilde{o}_{j}| + |\tilde{o}_{j}-\hat{o}_{j}|$. Combining the truncation and concentration bounds, which hold simultaneously for all $j$, gives Eq.~\eqref{eq:spot-shadow-bound}.
\end{proof}

The same test-round shadow record can be reused for every observable in
a family chosen independently of the data, whether specified before or
after the measurements. Additional observables therefore require only
classical post-processing. Standard non-i.i.d.\ classical shadows measure every round and estimate observables of the trajectory-averaged state
\cite{zambrano2026classical}. Our protocol instead certifies the
use-round average while leaving those systems unmeasured. In both
settings, the guarantee applies to an average, not to individual
states.

\begin{figure*}[!t]
\centering
\subfloat[\label{fig:energy_different_p}]{%
    \includegraphics[width=0.48\textwidth]{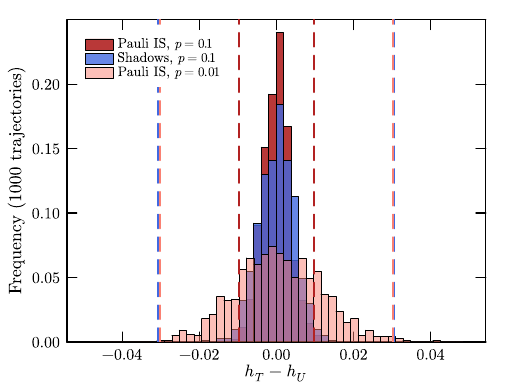}
}
\hfill
\subfloat[ \label{fig:many_observables}]{%
    \includegraphics[width=0.48\textwidth]{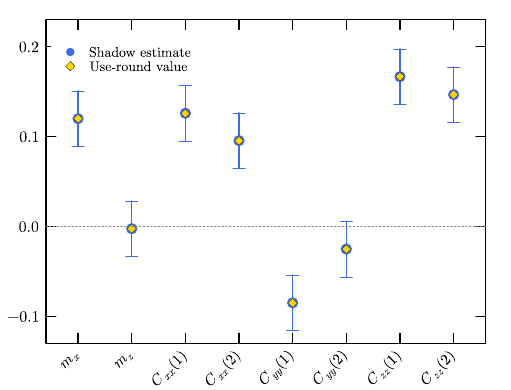}
}
\caption{From energy estimation to simultaneous certification with classical shadows under non-i.i.d.\ drift.
The target is the 3-qubit ground state of $H_{\rm TFIM}$ evolved under the drift model of Section~\ref{sec:drift}. Results are shown for 1000 independent trajectories with $N=10^7$ total rounds and $\delta=0.05$.
(a) The histogram shows the certification error $h_T-h_U$, where $h_T$ is the estimate obtained from the tested states, while $h_U$ is the exact average energy of the unmeasured states.
Blue and red correspond to the classical shadows and Pauli importance sampling (IS) certification method, respectively. Dashed vertical lines indicate the mean certification radii from Eqs.~\eqref{eq:pauli-freedman-radius} and~\eqref{eq:spot-shadow-bound} for the respective protocols. (b) Selected magnetizations and two-point correlators certified at $p=0.1$ from the same measurement data used to estimate the energy in (a).
Circles and diamonds show the estimated and exact use-round values, respectively, averaged over the trajectories. Error bars show the mean simultaneous certification radius from Eq. \eqref{eq:spot-shadow-bound}.
}
\label{fig:different_p}
\end{figure*}

The confidence bound in Theorem~\ref{thm:spot-shadows} depends on the observed number of use rounds $|S_U|$. To obtain a sufficient test bounds, suppose that
$p\leq1/2$ and $\epsilon \in (0, S]$, and let $K=pN$ be the expected number of test rounds. Apply Theorem~\ref{thm:spot-shadows} with
failure probability $\delta/2$, so that its logarithmic factor is $\ln(4L/\delta)$. On the event $|S_U|\geq qN/2$, the same calculation as in Corollary~\ref{cor:finite-certification-budget} bounds the last two terms in Eq.~\eqref{eq:spot-shadow-bound} by $4\epsilon/5$ whenever
\begin{align}\label{eq:shadow-test-complexity}
    K
    \geq
    \frac{25}{2q}
    \left(1+\frac{q}{6}\right)
    \frac{S^2}{\epsilon^2}
    \ln\!\left(\frac{4L}{\delta}\right).
\end{align}
This number of test rounds also ensures $\Pr(|S_U|<qN/2)\leq\delta/2$.
Combining these bounds with the truncation bias $\epsilon/5$, a union bound guarantees that $|S_U|>0$ and
$\max_j|\bar{o}_j-\hat{o}_j|\leq\epsilon$ with probability at least $1-\delta$. Thus, the target accuracy can be achieved with the familiar classical-shadow scaling
\begin{align}
    K = \mathcal{O}\!\left[\frac{S^2}{\epsilon^2} \ln\!\left(\frac{2L}{\delta}\right)
    \right].
\end{align}

\section{Numerical demonstrations} \label{sec:numerical}

The preceding sections establish bounds for both verification and certification. In the following numerical examples, we focus on the certification setting, where only a randomly selected subset of states is measured and the goal is to certify the average property of the unmeasured states. We first introduce a simple non-i.i.d.\ source model based on coherent drift, and then apply the certification protocol to representative observables arising in energy estimation and entanglement witnessing.

\subsection{History-dependent drift model} \label{sec:drift}
Many experiments exhibit gradual drift in state preparation due to phase, polarization, or wavelength fluctuations, as well as changes in calibration during a run
\cite{klimov2018fluctuations}. We model this behavior by applying a time-dependent coherent error to the target state $\sigma$. For each experimental realization, we draw a random Hermitian operator $H_{\rm drift}$ and keep it fixed throughout the trajectory. The error angle follows a Gaussian random walk, $\theta_t = \theta_{t-1} + \delta_t$, with independent increments $\delta_t \sim \mathcal{N}(0,\nu^2)$, where $\nu$ sets the typical angular change per round. The state prepared at round $t$ is
\begin{align}
    \rho_t = U_t \sigma U_t^\dagger,
    \qquad
    U_t = \exp(-i\theta_t H_{\rm drift}),
\end{align}
The sequence $\{\rho_1,\dots,\rho_N\}$ is non-i.i.d. since $\theta_t$ depends on $\theta_{t-1}$.

\subsection{From energy estimation to simultaneous certification}

We consider the certification of energy and local observables in many-body quantum simulations. First, if one knows in advance that only the energy is required, a tailored importance-sampling (IS) protocol certifies it efficiently. We focus on ground-state preparation for the one-dimensional transverse-field Ising model Hamiltonian
\begin{align}
    H_{\rm TFIM}
    =
    -J\left(\sum_i Z_i Z_{i+1} + g\sum_i X_i\right).
\end{align}
Here, $J>0$ is the exchange coupling and $g$ is the transverse-field strength. Since this Hamiltonian is a traceless Pauli sum, it can be used directly as the target observable, $W=H$. In this context, the true energy of any generated state $\rho_t$ is bounded by the spectrum of the Hamiltonian, such that $\omega_t \in [H_{\min}, H_{\max}]$.

To validate the protocol numerically, we apply it to a 3-qubit ground state preparation generated according to the drift model of Section~\ref{sec:drift}, with $\theta_0=0$ and $\nu=0.001$. The target state $\sigma$ is chosen to be the ground state of $H_{\rm TFIM}$, with parameters $J=1$ and $g=0.5$. We generate a sequence of $N=10^7$ total rounds, use a failure probability $\delta=0.05$, and repeat the certification protocol over 1000 independent trajectories. In each round, the state is tested with probability $p$. The test outcomes are used to construct the certification estimate
\begin{align}
    h_T
    =
    \frac{q}{p|S_U|}
    \sum_{t\in S_T} X_t,
\end{align}
where $X_t$ is the importance-sampled Pauli estimator in Eq.~\eqref{eq:single_shot_estimator_Pauli}. For benchmarking purposes, we also compute the exact average energy of the unmeasured states,
\begin{align}
    h_U
    =
    \frac{1}{|S_U|}
    \sum_{t\in S_U}\tr(H\rho_t),
\end{align}
which serves as the reference value.

Figure~\ref{fig:energy_different_p} shows, in red, the distribution of the certification error $h_T-h_U$ for importance-sampled Pauli measurements at test probabilities $p=0.1$ and $p=0.01$. The corresponding dashed lines denote the average confidence radius from Eq.~\eqref{eq:pauli-freedman-radius} across the different trajectories. In all cases, the
fraction of trials falling outside the confidence interval
remains below the chosen failure probability $\delta = 0.05$.
Specifically, we observe out-of-bound frequencies of at
most $0.7\%$ across the different methods.

If instead one wants a more detailed characterization of the state, the same spot-checking framework can be combined with classical shadows to certify many local observables simultaneously. In this approach, the shadow data collected on the test rounds allows us to simultaneously certify the average values of several observables over the unmeasured use states. To illustrate this, we use uniformly random local Pauli measurements and the truncated estimators of Section~\ref{sec:spot-shadows} to certify the energy together with the magnetization and two-point correlators
\begin{align}
    m_P
    &= \frac{1}{n}\sum_{i=1}^{n} P_i,
    \\
    C_{PP}(r)
    &= \frac{1}{n-r}\sum_{i=1}^{n-r} P_i P_{i+r},
\end{align}
where $P_i$ is any nontrivial Pauli operator acting on qubit $i$ (and identity elsewhere), and $r\in\{1,\dots,n-1\}$ is the distance between the two correlated qubits. This gives a family of $L=10$ observables, with a common confidence bound given by Eq.~\eqref{eq:spot-shadow-bound}. Note that since the observables are described by local Paulis, the shadow norm can be calculated exactly as shown in Appendix~\ref{app:exact-shadow-norm}. The single failure probability $\delta=0.05$ controls the entire family: with probability at least $0.95$, all ten estimates lie within the certified radius of their respective use-round averages (see Theorem \ref{thm:spot-shadows}).

Figure~\ref{fig:many_observables} compares the shadow estimates of selected magnetizations and correlators with their exact use-round values for $p=0.1$ and error $\epsilon\approx 5.4 \times 10^{-2}$. The plotted values are averaged over the 1000 trajectories, and the error bars show the mean simultaneous certification radius; the confidence guarantee applies within each trajectory. The same measurement data also provide the energy estimate shown in blue in Fig.~\ref{fig:energy_different_p}.

\subsection{Entanglement witnessing}

Another application of our protocol is entanglement certification.
We consider a source that initially prepares the partially entangled state
\begin{equation}
    \ket{\phi}
    =
    \frac{\sqrt{3}}{2}\ket{00}
    +
    \frac{1}{2}\ket{11},
\end{equation}
but may subsequently be affected by correlated noise or drift.
Since a fixed witness may cease to detect entanglement as the state
changes, we consider a family of witnesses
\begin{align}
    \mathcal{W}_j
    &= \frac{3}{4}\mathds{1}-\ketbra{\phi_j},
    \qquad
    \ket{\phi_j}
    = (C_A^{(j)}\otimes C_B^{(j)})\ket{\phi},
\end{align}
where $C_A^{(j)}$ and $C_B^{(j)}$ are single-qubit Clifford gates.
Taking all combinations of local Clifford gates yields $L=144$
distinct witnesses. Each satisfies
$\tr(\mathcal{W}_j\sigma)\geq0$ for every separable state $\sigma$,
so a negative expectation value certifies entanglement.
Since these witnesses are not traceless, we decompose them as
$\mathcal{W}_j=\alpha_I\mathds{1}+W_j$ and estimate the expectation
values of their traceless components $W_j$.

\begin{figure}[b]
\centering
    \includegraphics[width=1.0\columnwidth]{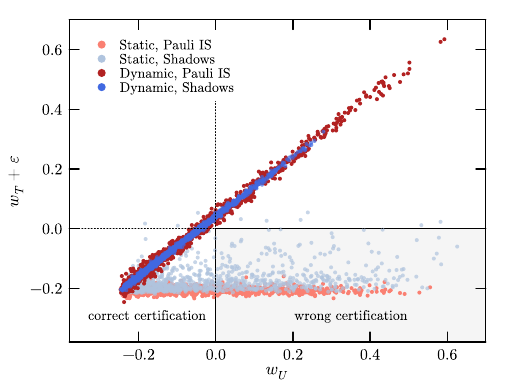}
\caption{
Entanglement-witness certification under source drift. Each data point represent one of the 1000 independent trajectories of a protocol with $N=5 \times 10^5$, and $\delta=0.05$. We plot the upper bound $w_T +\varepsilon$ against the exact use-round witness value $w_U$, where $w_T$ is the estimated witness value and $\varepsilon$ the confidence error associated with each protocol. Static protocols use $\epsilon = 0.04$, so we choose dynamic probabilities $p=0.008$ (Pauli IS) and $p=0.12$ (classical shadows) to match this error approximately. The shaded area represents wrong certification as the estimator claims entanglement certification being inconsistent with the nonnegative exact value $\omega_U$.
}
\label{fig:witness}
\end{figure}

We compare our Pauli IS and classical shadows certification
protocols with a static certification approach using the source model
of Section~\ref{sec:drift}, with initial phase $\theta_0=0$ and drift
rate $\nu=0.001$. In the static scheme, an initial block of states is
consumed to estimate the witness using the verification protocol of
Section~\ref{sec: verification protocol} for Pauli IS and the protocol
of Ref.~\cite[Theorem 2]{zambrano2026classical} for classical shadows. This estimate is then assumed to apply to the remaining states, which are used in a subsequent quantum task (e.g., quantum teleportation). We run each protocol on 1000 trajectories over $N=5 \times 10^5$ rounds and fix the failure probability to $\delta=0.05$. For our dynamic Pauli IS protocol, we choose the testing probability $p=0.008$, while for classical shadows we use $p=0.12$, so that the average certification errors from Eqs.~\eqref{eq:pauli-freedman-radius} and~\eqref{eq:spot-shadow-bound} approximately match the static target error $\epsilon=0.04$.

Figure~\ref{fig:witness} shows the entanglement claims made by
the four protocols. We plot the upper bound $w_T+\varepsilon$
against the exact average witness value over the use rounds,
$w_U$, where $w_T$ is the estimated witness value and $\varepsilon$
denotes the confidence error of each protocol. For both classical
shadow protocols, we choose $j^*=\arg\min_j\hat{o}_j$ and set
$w_T=\hat{o}_{j^*}$. For the Pauli IS protocols, we use the
unrotated witness $W_1$.

In the static approach, the initial testing block gives an optimistic
estimate of the witness value that becomes unrepresentative of the
later use states as the source drifts. The static protocols declare
entanglement in over 99\% of the runs. However, under drift, 34.1\%
of the runs for static Pauli IS and 38\% for static classical shadows
fall in the shaded region. Here, the inferred upper bound is negative,
while the exact use-round average satisfies $w_U\geq0$. We call these
incorrect certifications because the selected witness does not support
the entanglement claim for the use states.

In contrast, our dynamic certification protocol samples the trajectory
throughout the run. Dynamic Pauli IS correctly certifies entanglement
in 61.6\% of the runs, while dynamic classical shadows do so in 67.5\%,
with no incorrect certifications observed for either method.
Using the unrotated witness $W_1$ with the same
dynamic shadow data and confidence bound would certify entanglement
in 62.4\% of the runs. Selecting among the rotated witnesses raises
this fraction without additional measurements.

\section{Experimental Bell-state certification} \label{sec:exp}

\begin{figure}[t]
\centering
\includegraphics[width=1.0\columnwidth]{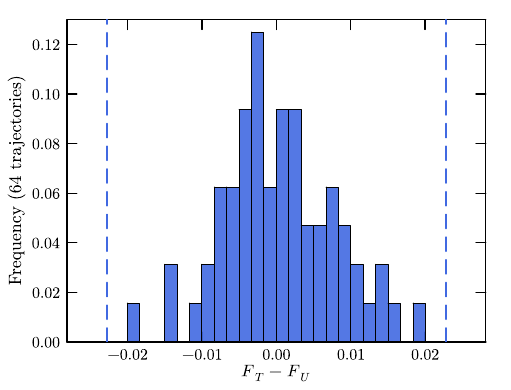}
\caption{
Experimental Bell-state certification on an IBM quantum processor.
The protocol certifies the average fidelity with respect to $\ket{\phi^+}=(\ket{00}+\ket{11})/\sqrt{2}$ on 64 disjoint connected physical qubit pairs.
For each pair, $N=815000$ Bell states are prepared, with each round assigned independently to the test set with probability $p=0.01$.
The figure shows $F_U-F_T$, where $F_T$ is the spot-checking certification estimate obtained from the tested states and $F_U$ is a separately measured benchmark estimate of the average fidelity of the use-round states.
The benchmark uncertainty is smaller than the certification radius and is not included in the dashed lines.
Dashed vertical lines indicate the average variance-sensitive certification bound in Eq.~\eqref{eq:pauli-freedman-radius}, with $\delta=0.05$.}

 \label{fig:DFE}
\end{figure}

We implement the certification protocol on an IBM quantum processor. We certify the fidelity with respect to the Bell state $\ket{\phi^+}=(\ket{00}+\ket{11})/\sqrt{2}$.  Following the convention of Section~\ref{sec:application}, the randomized measurement protocol is applied to the traceless part
\begin{align}
    W
    =
    |\phi^+\rangle \langle \phi^+|-\frac{1}{4}\mathds{1}
    =
    \frac{1}{4}
    \left(
    XX
    -
    YY
    +
    Z Z
    \right).
\end{align}
Thus $W_{\min}=-1/4$, $W_{\max}=3/4$, $\|\alpha\|_1=3/4$, and $\alpha_I = 1/4$.

To maximize data collection, we simultaneously prepare Bell states on 64 disjoint connected physical qubit pairs of the IBM quantum processor $\mathrm{ibm\_fez}$, treating each qubit pair as a parallel realization of the certification protocol. This implementation provides a natural scenario for our analysis, since superconducting quantum processors provide a realistic setting in which the i.i.d. assumption cannot be taken for granted. Parallel Bell-state preparations may be affected by residual cross-talk, correlated control errors, calibration drift, and other hardware imperfections \cite{harper2025crosstalk}. Such effects are treated as part of the underlying noise process and are naturally incorporated into the certification analysis.

In every round, each prepared Bell pair is independently assigned to a test round with probability $p=0.01$, and to a use round otherwise. As explained in Section~\ref{sec:complexity}, the certification protocol utilizes only the data collected in the test rounds to estimate the average value of the target property for the use-round states, which in this case is the fidelity. In an actual application, these use rounds would be consumed by a subsequent quantum task, such as quantum teleportation, entanglement swapping, or entanglement-based quantum key distribution.

For validation purposes only, we measure the states assigned to the use rounds and estimate their average fidelity. The certification protocol does not require this additional characterization; it serves solely as a benchmark against which the test round certification estimate can be compared. Since approximately $99\%$ of the prepared states are assigned to use rounds, the average statistical error of the benchmark estimate is $\epsilon \approx 2.3 \times 10^{-3}$, compared with $\epsilon \approx 2.3 \times 10^{-2}$ for the certification estimate obtained from the test rounds, given by Eq. \eqref{eq:pauli-freedman-radius}. Consequently, the benchmark fidelity provides a high-precision reference for evaluating the performance of the certification protocol.

Figure~\ref{fig:DFE} shows the difference between the benchmark use-round fidelity $F_U$ and the certification estimate $F_T$. The dashed lines indicate the average variance-sensitive certification error predicted by Eq.~\eqref{eq:pauli-freedman-radius}. The protocol was implemented on 64 physical qubit pairs using $N=815000$ rounds per pair and failure probability $\delta=0.05$. The observed agreement between $F_U$ and $F_T$ is consistent with the certification guarantee: the average fidelity of the use-round states is accurately predicted using only the information extracted from the tested subset.

\section{Conclusions} \label{sec:conclusion}

We have developed a framework for quantum state verification and certification that remains valid when successive preparations are neither independent nor identically distributed. Each state may depend arbitrarily on the preceding experimental history. The framework requires trusted measurements, with measurement settings sampled independently of that history, and, in certification, a private test decision sampled in the same way. Under these conditions, martingale methods provide rigorous confidence intervals for observable values averaged over the states generated during the experiment.

In full verification, all states are measured to estimate the average
value of a target observable over the entire trajectory. In spot-checking
certification, only a randomly selected subset is measured, and the
resulting confidence interval bounds the average value of the same
observable over the remaining systems. These systems can be used
immediately in a parallel quantum task, with their average observable
value certified by the test data.

We applied this construction to importance-sampled Pauli measurements and classical shadows. For Pauli measurements, covering energy estimation, direct fidelity estimation, and entanglement witnessing, full verification recovers the i.i.d.\ scaling $N=\mathcal{O}\!\left[\|\alpha\|_1^2\epsilon^{-2}\ln(1/\delta)\right]$. In the sparse-testing regime, the same leading dependence governs the expected number $pN$ of test rounds needed for certification. Spot-checking classical shadows allow the same test data to certify a finite family of $L$ observables simultaneously, retaining the characteristic scaling $pN=\mathcal{O}\!\left[\max_j\|W_j\|_{\rm sh}^2\epsilon^{-2}\ln(L/\delta)\right]$. In Eqs.~\eqref{eq:sparse-K-scaling} and~\eqref{eq:shadow-test-complexity} we give an explicit sufficient number of test rounds.

Our numerical and experimental results illustrate the operational value of these guarantees. Under history-dependent coherent drift, interleaved spot checking accurately estimates observable averages over the use rounds. Classical shadows allow several observables to be certified from the same data and improve entanglement detection by selecting among multiple witnesses. The entanglement witness example also shows how an initial testing block can become systematically unrepresentative of the later states, leading to unsupported entanglement claims. We also implemented Bell-state certification on 64 physical qubit pairs of an IBM quantum processor with test probability $p=0.01$. The certification estimates agreed with separately measured use-round benchmarks, consistent with the predicted confidence bounds.

The present framework certifies averages of linear quantum-state properties. Under the general source model considered here, this restriction is fundamental: given only one copy of each state in an unrestricted sequence, trajectory averages of non-affine state functionals cannot be estimated with vanishing worst-case error~\cite{zambrano2026quantum}. Convex and concave properties remain partially accessible, as information about the average state can be converted into one-sided bounds through Jensen's inequality. Tighter or two-sided estimates of nonlinear properties require additional assumptions about the source or access to multiple copies of the same state \cite{ekert2002direct}.

Two extensions would make the protocol more responsive to experimental conditions. Confidence bounds that remain valid at every round would allow an experiment to stop once the desired accuracy is reached. Allowing the testing probability and measurement strategy to depend on earlier data could also direct measurements toward the properties that remain most uncertain. Beyond these adaptive extensions, applications to self-testing and quantum-process certification would broaden the range of devices that can be treated without independence assumptions.

The broader conclusion is that independence among successive preparations is not indispensable to rigorous quantum verification and certification. What is needed is a trusted round-by-round measurement model, private random choices made independently of the preceding history, and an estimator whose conditional expectation represents the target observable value. Martingale concentration inequalities then turn these ingredients into statistically efficient guarantees, even when tests are sparse and the source exhibits drift, memory, or history-dependent noise. 

The code supporting this work is available online~\cite{git}.

% \clearpage
\begin{acknowledgments}
M.N. thanks Miguel Frías Pérez and Luciano Pereira for insightful discussions. This work was supported by the Government of Spain (Severo Ochoa CEX2019-000910-S, FUNQIP and QEC4QEA-PCI2025-163167), the European Union (PASQuanS2.1, 101113690 and QEC4QEA, 101194322), Fundació Cellex, Fundació Mir-Puig, and Generalitat de Catalunya (CERCA program). M.N. acknowledges funding from the European Union’s Horizon Europe research and innovation programme under the Marie Skłodowska-Curie Grant Agreement No. 101081441.\\
During the preparation of this work, the authors used ChatGPT Sol 5.6 and Astra 6 to assist with proof development, manuscript drafting,  and numerical code implementation. The authors verified the mathematical arguments, tested the code, and reviewed the final text, and take full responsibility for the work.

\end{acknowledgments}

\bibliography{biblio}

@article{brunner2014bell,
  title = {Bell nonlocality},
  author = {Brunner, Nicolas and Cavalcanti, Daniel and Pironio, Stefano and Scarani, Valerio and Wehner, Stephanie},
  journal = {Rev. Mod. Phys.},
  volume = {86},
  issue = {2},
  pages = {419--478},
  numpages = {60},
  year = {2014},
  month = {Apr},
  publisher = {American Physical Society},
  doi = {10.1103/RevModPhys.86.419},
  url = {https://link.aps.org/doi/10.1103/RevModPhys.86.419}
}

@article{eisert2020quantum,
  title={Quantum certification and benchmarking},
  author={Eisert, Jens and Hangleiter, Dominik and Walk, Nathan and Roth, Ingo and Markham, Damian and Parekh, Rhea and Chabaud, Ulysse and Kashefi, Elham},
  journal={Nat. Rev. Phys.},
  volume={2},
  number={7},
  pages={382--390},
  year={2020},
  publisher={Nature Publishing Group UK London}, 
  doi={10.1038/s42254-020-0186-4}
}

@article{pirandola2020advances,
  title={Advances in quantum cryptography},
  author={Pirandola, Stefano and Andersen, Ulrik L and Banchi, Leonardo and Berta, Mario and Bunandar, Darius and Colbeck, Roger and Englund, Dirk and Gehring, Tobias and Lupo, Cosmo and Ottaviani, Carlo and others},
  journal={Adv. Opt. Photonics},
  volume={12},
  number={4},
  pages={1012--1236},
  year={2020},
  publisher={Optical Society of America},
  doi={10.1364/AOP.361502}
}

@article{cramer2010efficient,
  title={Efficient quantum state tomography},
  author={Cramer, Marcus and Plenio, Martin B and Flammia, Steven T and Somma, Rolando and Gross, David and Bartlett, Stephen D and Landon-Cardinal, Olivier and Poulin, David and Liu, Yi-Kai},
  journal={Nat. Commun.},
  volume={1},
  number={1},
  pages={149},
  year={2010},
  publisher={Nature Publishing Group UK London},
  doi={10.1038/ncomms1147}
}

@article{georgescu2014quantum,
  title = {Quantum simulation},
  author = {Georgescu, I. M. and Ashhab, S. and Nori, Franco},
  journal = {Rev. Mod. Phys.},
  volume = {86},
  issue = {1},
  pages = {153--185},
  numpages = {33},
  year = {2014},
  month = {Mar},
  publisher = {American Physical Society},
  doi = {10.1103/RevModPhys.86.153},
  url = {https://link.aps.org/doi/10.1103/RevModPhys.86.153}
}

@article{farhi2014quantum,
  title={A quantum approximate optimization algorithm},
  author={Farhi, Edward and Goldstone, Jeffrey and Gutmann, Sam},
  journal={arXiv:1411.4028},
  year={2014},
  url={https://doi.org/10.48550/arXiv.1411.4028}
}

@article{hradil1997quantum,
  title = {Quantum-state estimation},
  author = {Hradil, Z.},
  journal = {Phys. Rev. A},
  volume = {55},
  issue = {3},
  pages = {R1561--R1564},
  numpages = {0},
  year = {1997},
  month = {Mar},
  publisher = {American Physical Society},
  doi = {10.1103/PhysRevA.55.R1561},
  url = {https://link.aps.org/doi/10.1103/PhysRevA.55.R1561}
}

@article{guta2020fast,
  title = {Fast state tomography with optimal error bounds},
  author = {Guţă, M and Kahn, J and Kueng, R and Tropp, J A},
  journal = {J. Phys. A: Math. Theor},
  volume = {53},
  number = {20},
  pages = {204001},
  year = {2020},
  publisher = {IOP Publishing},
  doi = {10.1088/1751-8121/ab8111}
}

@article{anshu2024survey,
  title={A survey on the complexity of learning quantum states},
  author={Anshu, Anurag and Arunachalam, Srinivasan},
  journal={Nat. Rev. Phys.},
  volume={6},
  number={1},
  pages={59--69},
  year={2024},
  publisher={Nature Publishing Group UK London},
  doi={10.1038/s42254-023-00662-4}
}

@article{huang2020predicting,
  title={Predicting many properties of a quantum system from very few measurements},
  author={Huang, Hsin-Yuan and Kueng, Richard and Preskill, John},
  journal={Nat. Phys.},
  volume={16},
  number={10},
  pages={1050--1057},
  year={2020},
  publisher={Nature Publishing Group UK London},
  doi= {10.1038/s41567-020-0932-7}
}

@article{flammia2011direct,
  title = {Direct Fidelity Estimation from Few {P}auli Measurements},
  author = {Flammia, Steven T. and Liu, Yi-Kai},
  journal = {Phys. Rev. Lett.},
  volume = {106},
  issue = {23},
  pages = {230501},
  numpages = {4},
  year = {2011},
  month = {Jun},
  publisher = {American Physical Society},
  doi = {10.1103/PhysRevLett.106.230501},
  url = {https://link.aps.org/doi/10.1103/PhysRevLett.106.230501}
}

@article{da2011practical,
  title = {Practical Characterization of Quantum Devices without Tomography},
  author = {da Silva, Marcus P. and Landon-Cardinal, Olivier and Poulin, David},
  journal = {Phys. Rev. Lett.},
  volume = {107},
  issue = {21},
  pages = {210404},
  numpages = {5},
  year = {2011},
  month = {Nov},
  publisher = {American Physical Society},
  doi = {10.1103/PhysRevLett.107.210404},
  url = {https://link.aps.org/doi/10.1103/PhysRevLett.107.210404}
}

@article{miguel2022collective,
  title = {Collective Operations Can Exponentially Enhance Quantum State Verification},
  author = {Miguel-Ramiro, Jorge and Riera-S\`abat, Ferran and D\"ur, Wolfgang},
  journal = {Phys. Rev. Lett.},
  volume = {129},
  issue = {19},
  pages = {190504},
  numpages = {6},
  year = {2022},
  month = {Oct},
  publisher = {American Physical Society},
  doi = {10.1103/PhysRevLett.129.190504},
  url = {https://link.aps.org/doi/10.1103/PhysRevLett.129.190504}
}

@article{riera2023nondestructive,
  title = {Nondestructive verification of entangled states via fidelity witnessing},
  author = {Riera-S\`abat, Ferran and Miguel-Ramiro, Jorge and D\"ur, Wolfgang},
  journal = {Phys. Rev. A},
  volume = {107},
  issue = {2},
  pages = {022414},
  numpages = {13},
  year = {2023},
  month = {Feb},
  publisher = {American Physical Society},
  doi = {10.1103/PhysRevA.107.022414},
  url = {https://link.aps.org/doi/10.1103/PhysRevA.107.022414}
}

@article{barbera2025sampling,
  doi = {10.22331/q-2025-07-03-1784},
  url = {https://doi.org/10.22331/q-2025-07-03-1784},
  title = {Sampling Groups of {P}auli Operators to Enhance Direct Fidelity Estimation},
  author = {Barber{\`{a}}-Rodr{\'{i}}guez, J{\'{u}}lia and Navarro, Mariana and Zambrano, Leonardo},
  journal = {{Quantum}},
  issn = {2521-327X},
  publisher = {{Verein zur F{\"{o}}rderung des Open Access Publizierens in den Quantenwissenschaften}},
  volume = {9},
  pages = {1784},
  month = jul,
  year = {2025}
}

@book{vershynin2018high, 
  title={High-Dimensional Probability: An Introduction with Applications in Data Science}, 
  publisher={Cambridge University Press}, 
  author={Vershynin, Roman}, 
  year={2018}, 
  doi={10.1017/9781108231596}
}

@article{schwonnek2021robust,
  title={Device-independent quantum key distribution with random key basis},
  author={Schwonnek, Ren{\'e} and Goh, Koon Tong and Primaatmaja, Ignatius W and Tan, Ernest Y-Z and Wolf, Ramona and Scarani, Valerio and Lim, Charles C-W},
  journal={Nat. Commun.},
  volume={12},
  number={1},
  pages={2880},
  year={2021},
  publisher={Nature Publishing Group UK London},
  doi={10.1038/s41467-021-23147-3}
}

@article{zhang2022device,
  title = {Device-Independent Quantum Key Distribution with Random Postselection},
  author = {Xu, Feihu and Zhang, Yu-Zhe and Zhang, Qiang and Pan, Jian-Wei},
  journal = {Phys. Rev. Lett.},
  volume = {128},
  issue = {11},
  pages = {110506},
  numpages = {5},
  year = {2022},
  month = {Mar},
  publisher = {American Physical Society},
  doi = {10.1103/PhysRevLett.128.110506},
  url = {https://link.aps.org/doi/10.1103/PhysRevLett.128.110506}
}

@article{white2022nonmarkovian,
  title = {Non-{M}arkovian Quantum Process Tomography},
  author = {White, G.A.L. and Pollock, F.A. and Hollenberg, L.C.L. and Modi, K. and Hill, C.D.},
  journal = {PRX Quantum},
  volume = {3},
  issue = {2},
  pages = {020344},
  numpages = {30},
  year = {2022},
  month = {May},
  publisher = {American Physical Society},
  doi = {10.1103/PRXQuantum.3.020344},
  url = {https://link.aps.org/doi/10.1103/PRXQuantum.3.020344}
}

@book{wiseman2009quantum,
  title={Quantum measurement and control},
  author={Wiseman, Howard M and Milburn, Gerard J},
  year={2009},
  publisher={Cambridge University Press},
  doi={10.1017/CBO9780511813948}
}

@article{pichler2017universal,
  title={Universal photonic quantum computation via time-delayed feedback},
  author={Pichler, Hannes and Choi, Soonwon and Zoller, Peter and Lukin, Mikhail D},
  journal={Proc. Natl. Acad. Sci. U.S.A.},
  volume={114},
  number={43},
  pages={11362--11367},
  year={2017},
  publisher={National Academy of Sciences},
  doi={10.1073/pnas.1711003114}
}

@article{klimov2018fluctuations,
  title={Fluctuations of energy-relaxation times in superconducting qubits},
  author={Klimov, Paul V and Kelly, Julian and Chen, Zijun and Neeley, Matthew and Megrant, Anthony and Burkett, Brian and Barends, Rami and Arya, Kunal and Chiaro, Ben and Chen, Yu and others},
  journal = {Phys. Rev. Lett.},
  volume={121},
  number={9},
  pages={090502},
  year={2018},
  publisher={APS},
  doi = {10.1103/PhysRevLett.121.090502}
}

@article{van2013quantum,
  title={When quantum tomography goes wrong: drift of quantum sources and other errors},
  author={van Enk, Steven J and Blume-Kohout, Robin},
  journal={New J. Phys.},
  volume={15},
  number={2},
  pages={025024},
  year={2013},
  publisher={IOP Publishing},
  doi={10.1088/1367-2630/15/2/025024}
}

@article{bavaresco2018measurements,
  title={Measurements in two bases are sufficient for certifying high-dimensional entanglement},
  author={Bavaresco, Jessica and Herrera Valencia, Natalia and Kl{\"o}ckl, Claude and Pivoluska, Matej and Erker, Paul and Friis, Nicolai and Malik, Mehul and Huber, Marcus},
  journal={Nat. Phys.},
  volume={14},
  number={10},
  pages={1032--1037},
  year={2018},
  publisher={Nature Publishing Group UK London},
  doi={10.1038/s41567-018-0203-z}
}

@article{Mannalath2025Sharp,
  title = {Sharp Finite Statistics for Quantum Key Distribution},
  author = {Mannalath, Vaisakh and Zapatero, V\'{\i}ctor and Curty, Marcos},
  journal = {Phys. Rev. Lett.},
  volume = {135},
  issue = {2},
  pages = {020803},
  numpages = {7},
  year = {2025},
  month = {Jul},
  publisher = {American Physical Society},
  doi = {10.1103/l735-x48g},
  url = {https://link.aps.org/doi/10.1103/l735-x48g}
}

@article{fawzi2024learning,
  title={Learning properties of quantum states without the IID assumption},
  author={Fawzi, Omar and Kueng, Richard and Markham, Damian and Oufkir, Aadil},
  journal={Nat. Commun.},
  volume={15},
  number={1},
  pages={9677},
  year={2024},
  publisher={Nature Publishing Group UK London},
  doi={10.1038/s41467-024-53765-6}
}

@article{zambrano2026quantum,
  title={Quantum tomography for non-iid sources},
  author={Zambrano, Leonardo},
  journal={arXiv:2602.22057},
  year={2026},
  url={https://doi.org/10.48550/arXiv.2602.22057
}
}

@article{zambrano2026classical,
  title={Classical shadows for non-iid quantum sources},
  author={Zambrano, Leonardo},
  journal={arXiv:2603.05137},
  year={2026},
  url={https://doi.org/10.48550/arXiv.2603.05137}
}

@article{dirkse2020witnessing,
  title={Witnessing entanglement in experiments with correlated noise},
  author={Dirkse, Bas and Pompili, Matteo and Hanson, Ronald and Walter, Michael and Wehner, Stephanie},
  journal={Quantum Sci. Technol.},
  volume={5},
  number={3},
  pages={035007},
  year={2020},
  publisher={IOP Publishing},
  doi={10.1088/2058-9565/ab8d88}
}

@article{neven2021symmetry,
  title={Symmetry-resolved entanglement detection using partial transpose moments},
  author={Neven, Antoine and Carrasco, Jose and Vitale, Vittorio and Kokail, Christian and Elben, Andreas and Dalmonte, Marcello and Calabrese, Pasquale and Zoller, Peter and Vermersch, Benoȋt and Kueng, Richard and others},
  journal={npj Quantum Inf.},
  volume={7},
  number={1},
  pages={152},
  year={2021},
  publisher={Nature Publishing Group UK London},
  doi={10.1038/s41534-021-00487-y}
}

@article{gocanin2022sample,
  title = {Sample-Efficient Device-Independent Quantum State Verification and Certification},
  author = {Go\ifmmode \check{c}\else \v{c}\fi{}anin, Aleksandra and \ifmmode \check{S}\else \v{S}\fi{}upi\ifmmode \acute{c}\else \'{c}\fi{}, Ivan and Daki\ifmmode \acute{c}\else \'{c}\fi{}, Borivoje},
  journal = {PRX Quantum},
  volume = {3},
  issue = {1},
  pages = {010317},
  numpages = {15},
  year = {2022},
  month = {Feb},
  publisher = {American Physical Society},
  doi = {10.1103/PRXQuantum.3.010317},
  url = {https://link.aps.org/doi/10.1103/PRXQuantum.3.010317}
}

@article{antesberger2026experimental,
  title={Experimental quantum state certification by actively sampling photonic entangled states},
  author={Antesberger, Michael and Schmid, Mariana ME and Cao, Huan and Daki{\'c}, Borivoje and Rozema, Lee A and Walther, Philip},
  journal={Sci. Adv.},
  volume={12},
  number={7},
  pages={eaea4144},
  year={2026},
  publisher={American Association for the Advancement of Science},
  doi={10.1126/sciadv.aea4144}
}

@article{knill2020generation,
  title = {Generation of quantum randomness by probability estimation with classical side information},
  author = {Knill, Emanuel and Zhang, Yanbao and Bierhorst, Peter},
  journal = {Phys. Rev. Res.},
  volume = {2},
  issue = {3},
  pages = {033465},
  numpages = {41},
  year = {2020},
  month = {Sep},
  publisher = {American Physical Society},
  doi = {10.1103/PhysRevResearch.2.033465},
  url = {https://link.aps.org/doi/10.1103/PhysRevResearch.2.033465}
}

@article{zhang2026efficient,
  title={An efficient method for spot-checking quantum properties with sequential trials},
  author={Zhang, Yanbao and Seshadri, Akshay and Knill, Emanuel},
  journal={arXiv:2602.08114},
  year={2026},
  url={https://doi.org/10.48550/arXiv.2602.08114}
}

@article{shah2024introductionmartingales,
      title={Introduction to Martingales}, 
      author={Rohan Shah},
      year={2024},
      journal={arXiv:2407.11914},
      url={https://arxiv.org/abs/2407.11914}, 
}

@article{fawzi2026optimal,
  title = {Optimal Fidelity Estimation from Binary Measurements for Discrete and Continuous Variable Systems},
  author = {Fawzi, Omar and Oufkir, Aadil and Salzmann, Robert},
  journal = {PRX Quantum},
  volume = {7},
  issue = {1},
  pages = {010309},
  numpages = {32},
  year = {2026},
  month = {Jan},
  publisher = {American Physical Society},
  doi = {10.1103/qd1c-1fk9},
  url = {https://link.aps.org/doi/10.1103/qd1c-1fk9}
}

@book{williams1991probability,
  title={Probability with martingales},
  author={Williams, David},
  year={1991},
  publisher={Cambridge University Press},
  doi={10.1017/CBO9780511813658}
}

@article{raginsky2015concentration,
      title={Concentration of Measure Inequalities in Information Theory, Communications and Coding (Second Edition)}, 
      author={Maxim Raginsky and Igal Sason},
      year={2012},
      journal={arXiv:1212.4663},
      url={https://arxiv.org/abs/1212.4663}, 
}

@article{horvitz1952generalization,
  title={A generalization of sampling without replacement from a finite universe},
  author={Horvitz, Daniel G and Thompson, Donovan J},
  journal={J. Am. Stat. Assoc.},
  volume={47},
  number={260},
  pages={663--685},
  year={1952},
  publisher={Taylor \& Francis},
  doi={10.1080/01621459.1952.10483446}
}

@article{harper2025crosstalk,
  title={Crosstalk attacks and defence in a shared quantum computing environment},
  author={Harper, Benjamin and Tonekaboni, Behnam and Goldozian, Bahar and Sevior, Martin and Usman, Muhammad},
  journal={Adv. Quantum Technol.},
  volume={8},
  number={10},
  pages={e2500009},
  year={2025},
  publisher={Wiley Online Library},
  doi={10.1002/qute.202500009}
}

@misc{git,
  title = {Certifying quantum states without independence assumptions - {Github repository}},
  year = {2026},
  publisher = {GitHub},
  journal = {GitHub repository},
  url = {https://github.com/MarianaNvrr/Certifying-quantum-states-without-independence-assumptions}
}

@article{hansen2003jensen,
author = {Hansen, Frank and Pedersen, Gert K.},
title = {Jensen's Operator Inequality},
journal = {Bull. Lond. Math. Soc.},
volume = {35},
number = {4},
pages = {553-564},
doi = {10.1112/S0024609303002200},
year = {2003}
}

@article{azuma1967weighted,
  title={Weighted sums of certain dependent random variables},
  author={Azuma, Kazuoki},
  journal={Tohoku Math. J.},
  volume={19},
  number={3},
  pages={357--367},
  year={1967},
  publisher={Mathematical Institute, Tohoku University}, 
  doi={10.2748/tmj/1178243286}
}

@article{freedman1975tail,
  title={On tail probabilities for martingales},
  author={Freedman, David A},
  journal={Ann. Probab.},
  volume = {3},
  number = {1}, 
  pages={100--118},
  year={1975},
  publisher={JSTOR},
  doi={10.1214/aop/1176996452}
}

@article{wecker2015progress,
  title = {Progress towards practical quantum variational algorithms},
  author = {Wecker, Dave and Hastings, Matthew B. and Troyer, Matthias},
  journal = {Phys. Rev. A},
  volume = {92},
  issue = {4},
  pages = {042303},
  numpages = {10},
  year = {2015},
  month = {Oct},
  publisher = {American Physical Society},
  doi = {10.1103/PhysRevA.92.042303},
  url = {https://link.aps.org/doi/10.1103/PhysRevA.92.042303}
}

@article{crawford2021efficient,
  doi = {10.22331/q-2021-01-20-385},
  url = {https://doi.org/10.22331/q-2021-01-20-385},
  title = {Efficient quantum measurement of {P}auli operators in the presence of finite sampling error},
  author = {Crawford, Ophelia and Straaten, Barnaby van and Wang, Daochen and Parks, Thomas and Campbell, Earl and Brierley, Stephen},
  journal = {{Quantum}},
  issn = {2521-327X},
  publisher = {{Verein zur F{\"{o}}rderung des Open Access Publizierens in den Quantenwissenschaften}},
  volume = {5},
  pages = {385},
  month = jan,
  year = {2021}
}

@article{ekert2002direct,
  title = {Direct Estimations of Linear and Nonlinear Functionals of a Quantum State},
  author = {Ekert, Artur K. and Alves, Carolina Moura and Oi, Daniel K. L. and Horodecki, Micha\l{} and Horodecki, Pawe\l{} and Kwek, L. C.},
  journal = {Phys. Rev. Lett.},
  volume = {88},
  issue = {21},
  pages = {217901},
  numpages = {4},
  year = {2002},
  month = {May},
  publisher = {American Physical Society},
  doi = {10.1103/PhysRevLett.88.217901},
  url = {https://link.aps.org/doi/10.1103/PhysRevLett.88.217901}
}

@article{Bancal2021selftestingfinite,
  doi = {10.22331/q-2021-03-02-401},
  url = {https://doi.org/10.22331/q-2021-03-02-401},
  title = {Self-testing with finite statistics enabling the certification of a quantum network link},
  author = {Bancal, Jean-Daniel and Redeker, Kai and Sekatski, Pavel and Rosenfeld, Wenjamin and Sangouard, Nicolas},
  journal = {{Quantum}},
  issn = {2521-327X},
  publisher = {{Verein zur F{\"{o}}rderung des Open Access Publizierens in den Quantenwissenschaften}},
  volume = {5},
  pages = {401},
  month = mar,
  year = {2021}
}

\appendix
\section{Exact shadow norm for local Pauli measurements}
\label{app:exact-shadow-norm}

We compute the shadow norm (see Eq.~\eqref{eq:shadow-norm}) for uniformly random local Pauli measurements. We recall the traceless observable \(W\) from Eq.~\eqref{eq:traceless operator W} 
and denote the support of \(P\) by
\(\operatorname{supp}(P)\), as well as the Pauli weight \(w(P) = | \operatorname{supp}(P)|\), which represents the number of qubits on which the Pauli string acts nontrivially.

Here the measurement setting is the string \(b\in\{X,Y,Z\}^n\), sampled uniformly, and its
outcome is \(s\in\{+1,-1\}^n\) (note that in the main text the index $k = (b,s)$ denotes the complete measurement record). Defining
\begin{align}
    \Pi_{s_i}^{b_i}
    =
    \frac{\mathds{1}+s_i b_i}{2},
\end{align}
the corresponding POVM element and classical snapshot are
\begin{align}\label{eq:local-pauli-shadow}
    E_{b,s}
    &=
    \frac{1}{3^n}\bigotimes_{i=1}^n\Pi_{s_i}^{b_i},
    &
    \hat{\rho}_{b,s}
    &=
    \bigotimes_{i=1}^n
    \left(3\Pi_{s_i}^{b_i}-\mathds{1}\right).
\end{align}
For a Pauli string \(P\), its single-shot score satisfies
\begin{align}\label{eq:pauli-shadow-score}
    \tr\!\left(P\hat{\rho}_{\bm{b},\bm{s}}\right)
    =
    3^{w(P)}
    \prod_{i\in\operatorname{supp}(P)}s_i \delta_{b_i, P_i}.
\end{align}

Two Pauli strings are called compatible, denoted by \(P\sim P'\), if
\(P_i=P_i'\) for every $i$ on which both strings act nontrivially.
Identifying $X(b,s) =  \tr(W \hat{\rho}_{b,s})$ in
Eq.~\eqref{eq:second-moment-operator} and using
Eq.~\eqref{eq:pauli-shadow-score}, the second moment operator is
\begin{align}\label{eq:exact-pauli-shadow-operator}
    Q_W
    &=\sum_{\substack{
        P,P'\in\mathcal{P}_n\setminus\{I\}
    }}
    \alpha_P\alpha_{P'}\sum_{b,s}
    \tr\!\left(P\hat{\rho}_{b,s}\right)
    \tr\!\left(P'\hat{\rho}_{b,s}\right)
    E_{b,s}\nonumber \\
    &=\sum_{\substack{
        P,P'\in\mathcal{P}_n\setminus\{I\}\\
        P\sim P'
    }}
    \alpha_P\alpha_{P'}
    \frac{3^{w(P)+ w(P')}}{3^{|\operatorname{supp}(P)\cup\operatorname{supp}(P')|}}
    PP' \nonumber \\
    &=
    \sum_{\substack{
        P,P'\in\mathcal{P}_n\setminus\{I\}\\
        P\sim P'
    }}
    \alpha_P\alpha_{P'}
    3^{|\operatorname{supp}(P)\cap\operatorname{supp}(P')|}
    PP'.
\end{align}
Indeed, an incompatible pair cannot be matched by a common local
measurement setting and therefore contributes zero. 

Finally, \(Q_W\) is positive semidefinite because it is a sum of POVM
elements with nonnegative weights. If
\(Q_W=\sum_r q_r\ketbra{r}\) is its spectral decomposition, then for
every density matrix \(\sigma\),
\begin{align}
    \tr(\sigma Q_W)
    =
    \sum_r q_r\langle r|\sigma|r\rangle
    \leq
    \max_r q_r,
\end{align}
since \(\langle r|\sigma|r\rangle\geq0\) and these numbers sum to one.
Equality is attained for
\(\sigma=\ketbra{r_{\max}}\). Consequently,
\begin{align}\label{eq:exact-pauli-shadow-norm}
        \norm{W}_{\rm sh}^2
        =
        \lambda_{\max}(Q_W).
\end{align}
Thus, the exact squared shadow norm is obtained by collecting terms that produce the same Pauli product in Eq.~\eqref{eq:exact-pauli-shadow-operator} and computing the largest eigenvalue of the resulting second-moment operator $Q_W$.  For large numbers of qubits, explicitly constructing $Q_W$ and computing its largest eigenvalue can become computationally expensive.  A scalable alternative is to bound each compatible
pairwise contribution in absolute value, leading to an analytical upper bound ~\cite[Corollary 2]{huang2020predicting, zambrano2026classical}.
\end{document}